\documentclass[pra,twocolumn,showpacs]{revtex4-2}

\usepackage{amsmath}
\usepackage{latexsym}
\usepackage{amssymb}
\usepackage{graphicx}
\usepackage[colorlinks=true, citecolor=blue, urlcolor=blue]{hyperref}
\usepackage{float}
\usepackage{amsfonts}
\usepackage{textcomp}
\usepackage{mathpazo}
\usepackage{comment}
\usepackage{xr}

\usepackage{bbm}

\usepackage{xcolor}
\definecolor{myurlcolor}{rgb}{0,.5,.5}
\definecolor{mycitecolor}{rgb}{0,.6,0}
\definecolor{myrefcolor}{rgb}{2,0,0}
\usepackage{hyperref}
\hypersetup{colorlinks,
linkcolor=myrefcolor,
citecolor=mycitecolor,
urlcolor=myurlcolor}

\usepackage[draft]{fixme}
\usepackage{empheq}
\usepackage{amsmath,bbm}%,mathrsfs}
\usepackage{graphicx}
\usepackage{amsfonts}
\usepackage{amssymb}
\usepackage{amsmath, amssymb, amsthm,verbatim,graphicx,bbm}
\usepackage{mathrsfs}
\usepackage{subcaption}
\usepackage{ragged2e}

\usepackage{xr}
\makeatletter
\newcommand*{\addFileDependency}[1]{
  \typeout{(#1)}
  \@addtofilelist{#1}
  \IfFileExists{#1}{}{\typeout{No file #1.}}
}
\makeatother

\newcommand*{\myexternaldocument}[1]{
    \externaldocument{#1}
    \addFileDependency{#1.tex}
    \addFileDependency{#1.aux}
}
\myexternaldocument{manuscript_PRL}
\newcommand{\beq}[0]{\begin{equation}}
\newcommand{\eeq}[0]{\end{equation}}

\newcommand{\one}{\leavevmode\hbox{\small1\normalsize\kern-.33em1}}

\def\be{\begin{equation}}
\def\ee{\end{equation}}
\def\ben{\begin{eqnarray}}
\def\een{\end{eqnarray}}
\def\eea{\end{array}}
\def\bea{\begin{array}}

\newcommand{\Tr}[1]{\mathrm{Tr}#1}
\newcommand{\bei}{\begin{itemize}}
\newcommand{\eei}{\end{itemize}}
\newcommand{\ket}[1]{|#1\rangle}
\newcommand{\bra}[1]{\langle#1|}

\newcommand{\proj}[1]{\ket{#1}\!\!\bra{#1}}

\newcommand{\I}{\mathbbm{1}}

\renewcommand{\emph}[1]{\textbf{#1}}

\makeatletter
\newtheorem*{rep@theorem}{\rep@title}
\newcommand{\newreptheorem}[2]{%
\newenvironment{rep#1}[1]{%
 \def\rep@title{#2 \ref{##1}}%
 \begin{rep@theorem}}%
 {\end{rep@theorem}}}
\makeatother

\theoremstyle{plain}
\newtheorem{thm}{Theorem}%[section]
\newtheorem*{thm*}{Theorem}
\newreptheorem{thm}{Theorem}

\newtheorem*{ucor}{Corollary}

\theoremstyle{definition}

\theoremstyle{remark}

\usepackage[T1]{fontenc}

\begin{document}

\title{One-sided-DIQKD using any pure bipartite entangled state secure against all attacks}
\author{Bohnishikha Ghosh}
\email{B.Ghosh@hw.ac.uk}
\affiliation{
Institute of Photonics and Quantum Sciences, Heriot-Watt University, Edinburgh EH14 4AS, UK}
\author{Chithra Raj}
\email{chithra.raj@ug.edu.pl}
\affiliation{International Centre for Theory of Quantum Technologies,
University of Gdańsk, prof. Marii Janion 4, 80-309 Gdańsk, Poland}

\author{Shubhayan Sarkar}
\email{shubhayan.sarkar@jcbose.ac.in}
\affiliation{Department of Physical Sciences, Bose Institute, EN 80, Bidhannagar, Kolkata 700091, India}

\begin{abstract}
Quantum Key Distribution (QKD) allows spatially separated parties to establish a secret key among themselves based on quantum principles. While QKD systems are being deployed worldwide, their reliance on trusted measurements makes them vulnerable. Moreover, existing protocols have largely focussed on the qubit cases, and rely on the maximally entangled state, a costly resource to generate. Here, we consider the one-sided-device-independent (1SDI) QKD protocols, where one of the measurement devices is trusted while the other is untrusted. The resource in this scenario is quantum steering, an asymmetric form of quantum nonlocality. This is inspired by a real-world scenario, where the trusted party can be considered to be a bank, or any large organisation that has sufficient resources to protect its systems and validate them, while the users are vulnerable to attacks.  We then demonstrate that a secure key can be obtained using any pure bipartite entangled state in this setup. Considering realistic detectors and sources based on current technology, we demonstrate that even states that are almost separable can be used to extract a secure key. Moreover, for every pure entangled bipartite state, we find the critical detection efficiencies and visibility to obtain a secure key. 
\end{abstract}

\maketitle

{\it{Introduction---}} With the advent of quantum computers, conventional public-key cryptographic schemes, such as RSA, face potential security threats. Quantum Key Distribution (QKD) has therefore emerged as a promising approach to securing communications against future quantum-enabled attacks, attracting growing interest worldwide~\cite{QKDrev1,QKDrev2,DIQKDrev3}. Unlike traditional key distribution protocols that are based on the hardness of computing problems and will be compromised by advances in algorithms and computing technologies, the security of QKD protocols is embedded in the laws of physics, offering universal information-theoretic security. The first QKD protocol, BB84~\cite{bennett1984quantum}, encodes information in non-orthogonal quantum states, exploiting the disturbance introduced by attempts to extract information about the encoded states to detect eavesdropping. This principle was subsequently extended to entanglement-based protocols, notably E91~\cite{Ekert1991E91}, which relates key security to Bell-inequality violations, and BBM92~\cite{bennett1992quantum}, the entanglement-based counterpart of BB84. Currently, QKD is being experimentally demonstrated and deployed in fibre-optic and satellite-based communication networks as a potential means of protecting sensitive information against quantum-enabled attacks~\cite{expQKD1,expQKD2,expQKD3,QKDep1,QKDep2,QKDep3,QKDep4,QKDep5,QKDep6,QKDep7,QKDep8,QKDep9,QKDep11}. 

However, the reliance of trust in measurements in the QKD protocols makes them vulnerable to simple attacks such as detector-blinding attacks \cite{attack1}, and subsequently many more atttacks~\cite{attack2,attack11,attack3,attack4,attack5,attack6,attack7,attack8,attack9,attack10}. This inspired protocols where the trust in the measurement devices can be relaxed. One such class of protocols are called one-sided-device-independent (1SDI)-QKD protocols, where one of the measurement devices is trusted, while the other is untrusted~\cite{1SDIQKD1,1SDIQKD2,1SDIQKD3,1SDIQKD4,1SDIQKD5}. The nonclassical resource in this scenario is known as quantum steering \cite{wiseman2006steering}, an asymmetric form of Bell nonlocality~\cite{Bell1964}. In practice, this represents a scenario where a large organisation with sufficient resources can regularly validate its systems to ensure no intruder is present, while users remain vulnerable to attacks. Recent demonstrations have shown that 1SDI-QKD protocols can achieve key rates and transmission distances comparable to those of conventional QKD, while operating with low detection efficiencies accessible with current technology~\cite{1SDIQKD4,1SDIQKD5}. Moreover, most of the existing protocols are based on the singlet, which is a  maximally entangled state, and therefore a valuable resource in the resource theory of entanglement~\cite{Horodecki2009}. Furthermore, it is well-established that the critical detection efficiency, that is, the minimal detection efficiency to observe Bell nonlocality, can be lower for nearly product states than for their maximally entangled counterparts~\cite{Eberhard1993,Giustina2015,Shalm2015}. Similar effects have been observed for quantum steering~\cite{Vallone_2013}.

In this work, we demonstrate that every bipartite entangled state is useful for 1SDI-QKD. Consequently, even states that are almost product are useful for generating a secure key between distant parties. Considering realistic sources and inefficient detectors, we find explicit key rates for any pure bipartite entangled state. We then find the critical detection efficiency for all such states, and observe that it is robust to entanglement; that is, even when the entanglement is substantially reduced from its maximal value, the critical detection efficiency decreases only gradually, thus establishing that high-quality maximally entangled states are not necessary for achieving a low threshold for detection efficiency. Moreover, our protocol is secure against the most general attacks, also known as coherent attacks.
Before proceeding to the main results, let us describe the scenario and the 1SDI-QKD protocol.

\begin{figure}
    \centering
    \includegraphics[width=\linewidth]{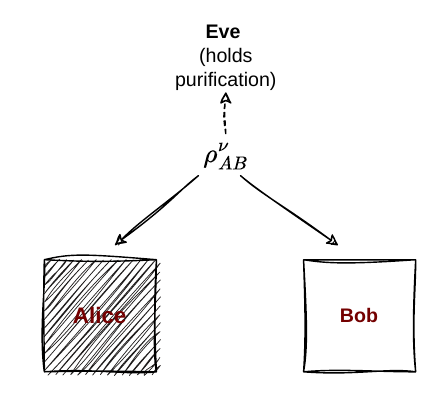}
    \caption{1SDI-QKD scenario: Alice is untrusted while Bob is trusted.}
    \label{fig:placeholder}
\end{figure}

{\it{1SDI-QKD---}}\label{sec:preliminaries}
We consider a one-sided device-independent (1SDI) scenario in which
an untrusted source distributes a bipartite quantum state $\rho_{AB}$
between two spatially separated parties, Alice and Bob. Alice's
measurement device is completely uncharacterised and is treated as a black box, while Bob's measurements are trusted and fully
characterised~\cite{branciard2012one}. This asymmetric trust structure corresponds naturally to the setting of quantum steering \cite{wiseman2006steering}. For each measurement choice $x$ and outcome $a$ of Alice's untrusted device, Bob is left with the subnormalized conditional state $ \sigma_{a|x}
    =
    \Tr_{B}\left[
        \left(M_{a|x}\otimes \I_B\right)\rho_{AB}
    \right],$
where $M_{a|x}$ denotes a possible quantum realisation of Alice's
measurement. The collection
$\{\sigma_{a|x}\}_{a,x}$ forms an assemblage on Bob's trusted
system. The corresponding probabilities that Alice obtains outcome $a$ by performing measurement $x$ are given by $  p(a|x)=\Tr[\sigma_{a|x}]$.
The assemblage satisfies $ \sum_a \sigma_{a|x}=\rho_B$ for all $x$, and is thus,
independent of Alice's input.

An assemblage is unsteerable if it admits a local-hidden-state (LHS)
decomposition of the form $  \sigma_{a|x}
    =
    \sum_{\lambda}
    p(\lambda)\,
    p(a|x,\lambda)\,
    \rho_{\lambda}^{B},$
where $\lambda$ is a classical hidden variable and
$\rho_{\lambda}^{B}$ are local quantum states on Bob's side. When LHS decomposition of the assemblage does not exist, Alice is said to steer Bob's system~\cite{wiseman2006steering}. In a
1SDI cryptographic setting, this allows correlations to be certified
without requiring any characterisation of Alice's measurement
device, while retaining a trusted quantum description on Bob's side.

We now describe the protocol considered throughout this work. We follow the 1SDI-QKD protocol in \cite{1SDIQKD4}, also called the 1SDI-QKD protocol with classical one-way reconciliation, in which classical communication is restricted from Alice to Bob and the key is extracted from Bob's measurement outcomes. Here, we consider that only a single basis corresponding to the input $y=0$ is used to extract the key. 
% {\color{blue}\textit{two-basis/spot-checking...what about other protocols?}} 
In each round, an untrusted source distributes a bipartite state $\rho_{AB}$
to Alice and Bob. Alice chooses a measurement setting
$x\in\mathcal{X}$ and records an outcome $a$, while Bob chooses a
setting $y\in\mathcal{Y}$ and performs a trusted POVM
$\{B_{b|y}\}_b$, obtaining an outcome $b$. The corresponding
joint statistics or correlations are given by $\vec{p}=\{p(a,b|x,y)\}$ where $ p(a,b|x,y)
    =
    \Tr\left[
        \left(M_{a|x}\otimes B_{b|y}\right)\rho_{AB}
    \right].$
The operators $M_{a|x}$ are not assumed to be known and merely
represent a possible realisation of Alice's black-box device. 
The protocol proceeds as follows. Alice and Bob first repeat the
state-distribution and measurement stage over many rounds. Then Alice randomly assigns some rounds for parameter estimation, that is, checking whether the expected statistics are attained, while reserving the other rounds for generating the key. For the parameter estimation rounds, Alice
publicly reveals her measurement choices and outcomes, based on which Bob checks the expected statistics. If the test passes, then Bob extracts the key using his outcomes ~\cite{1SDIQKD4,1SDIQKD5}. Alice then corrects her raw string to match Bob's during information reconciliation. This choice makes use of the asymmetry of the 1SDI setting, since Bob's measurement device is
trusted while Alice's is not. Privacy amplification is subsequently
performed to remove any information that may be available to an
eavesdropper. 

Since, we want to estabilish a secure key in a realistic scenario, we need to consider that the sources and the detectors are inefficient. Detection losses must be treated differently on the trusted and untrusted sides. Since Bob's measurements are trusted and fully characterised, losses on Bob's side can be ignored using the fair-sampling assumption \cite{fair-sampling}. On Alice's untrusted side, however, no-click
events cannot in general be discarded, since post-selecting on the outputs of an untrusted device can open the detection loophole~\cite{1SDIQKD3, 1SDIQKD4}. We therefore retain Alice's no-click events as an additional outcome and include them explicitly in the
observed statistics~\cite{1SDIQKD4, 1SDIQKD5}. 
\iffalse
The corresponding probabilities are therefore
\begin{equation}
    p_\eta(a|x)=\eta p(a|x),
    \qquad
    p_\eta(\emptyset|x)=1-\eta.
    \label{eq:lossy_probabilities}
\end{equation}
Thus, rather than conditioning the analysis on successful detection, the no-click events remain part of Alice's outcome alphabet.
\fi

As shown in \cite{1SDIQKD3, 1SDIQKD4}, this allows the protocol to be secure against the most general class of attacks in a QKD setup, called coherent attacks, as the protocol admits the sequential structure required by entropy-accumulation methods~\cite{dupuis2020entropy,metger2024generalised}. Eve may therefore be allowed to correlate different rounds, retain an arbitrary quantum memory, and postpone her measurement until after the public communication. The single-round entropy bounds derived below can, under the corresponding entropy-accumulation conditions, be extended
beyond collective attacks to coherent attacks~\cite{1SDIQKD4,metger2024generalised}. Moreover, to account for the presence of an adversary, we allow Eve to hold a
purification of the state shared by Alice and Bob, $ \ket{\Psi}_{ABE},
    \quad
    \rho_{AB}
    =
    \Tr_E\left[
        \ket{\Psi}\bra{\Psi}_{ABE}
    \right].$
    
In the asymptotic limit, the key-rate in the above described protocol is lower-bounded by the Devetak-Winter rate~\cite{devetak2005distillation} given by
\begin{equation}
    r_{\infty}
    \geq
    H(B|E)-H(B|A).
    \label{eq:DW}
\end{equation}
Here, $H(B|E)$ quantifies Eve's uncertainty about Bob's raw key,
while $H(B|A)$ accounts for the information required during error
correction. The central task in the security analysis is therefore to
lower bound $H(B|E)$ from the observed data without making assumptions
about the internal implementation of Alice's device. Let us first estabilish the optimal key rate in the noiseless case, which follows directly from self-testing of any bipartite entangled pure quantum state in the 1SDI scenario \cite{Sarkar_2023}. 

{\it{Key rate for any pure bipartite entangled state in the noiseless case---}}
Let us consider the 1SDI-QKD protocol described above, where, in the ideal realisation, Alice and Bob share an arbitrary pure entangled state of local dimension $d$. Up to a choice of local bases, any pure bipartite state can be written as
\begin{equation}
    \ket{\psi{(\boldsymbol{\alpha})}}_{AB} = \sum_{i=0}^{d-1}\alpha_i\ket{ii},
    \qquad \alpha_i=\alpha_i^*,\quad\sum_{i=0}^{d-1}\alpha_i^2=1.
    \label{eq:gen_state}
\end{equation}
Alice and Bob each choose between two measurement settings, $x,y\in\{0,1\}$. Alice's measurements are given by
\begin{eqnarray}\label{idealmea}
    Z_d=\{\proj{i}\}_{i=0}^{d-1}, \quad X_d=\{\proj{e_i}\}_{i=0}^{d-1}
\end{eqnarray}
where $\ket{e_j}=\frac{1}{\sqrt{d}}(\sum_{k=0}^{d-1}\omega^{jk}\ket{k})$ and $\omega=e^{2\pi i/d}$. Bob's measurements are $Z_d^*,X_d^*$. Now, one can directly compute the expected statistics $\vec{p}$ [see Eqs. \eqref{eq:pabZZ} and \eqref{eq:pabXX} for visibility $\nu=1$ and detection efficiency $\eta=1$]. 

As shown in \cite{Sarkar_2023}, the above correlations lead to self-testing of the quantum state as
\begin{eqnarray}\label{self-testst}
  (U_A\otimes\I_B) \ket{\Psi_{ABE}}=\ket{\psi{(\boldsymbol{\alpha})}}_{A'B}\otimes\ket{\xi}_{A''E},
\end{eqnarray}
while the measurements of Alice $\{A_0,A_1\}$ are certified as
\begin{eqnarray}\label{self-testmea}
    U_A A_0 U_A^{\dagger}=Z_d\otimes\I_{A''},\quad U_A A_1 U_A^{\dagger}=X_d\otimes\I_{A''},
\end{eqnarray}
where $U_A:\mathcal{H}_A\rightarrow \mathbb{C}^d\otimes\mathcal{H}_{A''}$ represents a unitary transformation and $\ket{\xi}_{A''E}$ represents the junk state that does not contribute to the protocol.
This allows us to directly obtain a lower bound on the key rate in the noiseless case.  

\begin{thm}
\label{thm:pure}
In the ideal realisation of the 1SDI-QKD protocol described above, let Alice and Bob share the pure bipartite state in~\eqref{eq:gen_state}
and perform the measurements in~\eqref{idealmea}. The lower bound to the key-rate for this protocol when the key is extracted from the input $y=0$ is given by
\begin{equation}
    r_{\infty} \geq S(\rho_B) = -\sum_{i=0}^{d-1} \alpha_i^2\log_2\alpha_i^2,
    \label{eq:pure_optimal_rate}
\end{equation}
where $\alpha_i$ are the parameters of the ideal state \eqref{eq:gen_state}.
\end{thm}

\begin{proof} As we observe from the Devetak-Winter rate \eqref{eq:DW}, let us first compute the conditional entropy $ H(B|A)$. For this purpose, the conditional joint probabilities are given by $ p(b\mid a,Z,Z^*)=\delta_{a,b}$, and thus the conditional entropy is simply $H(B|A,Z,Z^*)=0$. As the observed correlations allow us to self-test the quantum state and Alice's measurements, as stated in Eqs. \eqref{self-testst} and \eqref{self-testmea} \cite{Sarkar_2023}, the state shared between Bob and Eve is a product state given by
\begin{equation}
\rho_{BE}=\Tr_{A}(\Psi_{ABE})=\sum_{i=0}^{d-1}\alpha_i^2\proj{i}\otimes\Tr_{A''}(\xi_{A''E}).
\end{equation}
This allows us to conclude simply that $H(B|E)=H(B)$, as no action by Eve can alter Bob's subsystem. Since, $p(b|Z,Z^*)=\alpha_b^2$, we simply arrive at $r_{\infty}\geq H(B)=-\sum_{i=0}^{d-1} \alpha_i^2\log_2\alpha_i^2$.
\end{proof}
For any pure bipartite state, $S(\rho_B)$ is precisely the entropy of entanglement \cite{Horodecki2009}. Theorem~\ref{thm:pure} therefore shows that, in the noiseless 1SDI-QKD protocol,
the key rate is determined directly by the entanglement of
the shared state. In particular, maximal entanglement is not
necessary to obtain a non-zero key rate. Any pure entangled state has $S(\rho_B)>0$ and hence yields a key in the 1SDI-QKD protocol stated above.

{\it{1SDI-QKD key rate for any pure bipartite entangled state in realistic setups---}}
No physical source can generate the expected state perfectly, nor can any detector detect every possible subsystem. It is well established that an attacker can exploit these vulnerabilities to compromise the security of the protocol. Consequently, we consider the simplest model of incorporating noise in the state as
\begin{equation} \rho_{AB}^{\nu} = \nu\proj{\psi(\boldsymbol{\alpha})}+ (1-\nu)\frac{\I_{AB}}{d^2},
\label{eq:noisy_state}
\end{equation} 
where $\nu\in[0,1]$ is the visibility. Notice that the above state corresponds the original state in eq.~\eqref{eq:gen_state} in the presence of depolarising noise, which is a commonly considered noise model in the source \cite{Giustina2015,Shalm2015,QKDrev1,QKDrev2}. 

For the detectors, we again consider a simplistic but the most common way to treat losses of events. The loss enters through Alice's detection efficiency $\eta\in[0,1]$: a successful detection
yields an outcome $a\in\{0,\ldots,d-1\}$ and a no-click event is denoted $a=\emptyset$.
Bob's device is trusted and always returns an
outcome $b\in\{0,\ldots,d-1\}$.
More precisely, let $\{M_{a|x}\}_{a=0}^{d-1}$ denote an ideal measurement corresponding to Alice's setting $x$. For a detection efficiency $\eta$, we include an additional no-click outcome $\emptyset$ and model the lossy measurement as
\begin{equation}
    M_{a|x}^{(\eta)}
    =
    \eta M_{a|x},
    \qquad
    M_{\emptyset|x}^{(\eta)}
    =
    (1-\eta)\I_A.
    \label{eq:lossy_measurement}
\end{equation}
However, as Bob is trusted, one can discard the no-detection events and obtain the standard measurements $\{M_{b|y}\}$, also called fair-sampling \cite{1SDIQKD1,1SDIQKD4,1SDIQKD5,1SDIQKD3,1SDIQKD2}.

The correlations generated from the above ideal realisations do not self-test the quantum state or measurements as in Eqs. \eqref{self-testst} and \eqref{self-testmea}. Consequently, $H(B|E)\ne H(B)$ as Eve can now be correlated to Bob. To obtain a lower bound to $H(B|E)$, we consider the relation introduced in \cite{Berta2010}, which, when the measurements of Bob are mutually unbiased, are given by
\begin{eqnarray}
    H(B|E)\geq \log_2 d-H(B|A,X,X^*).
\end{eqnarray}
it is worth noting that the above bound is derived using the entropic-uncertainty relation in the presence of quantum memory.

Let us now construct the correlations that are expected by Alice and Bob from the parameter estimation rounds, which, when both choose inputs $x=y=0$, is
\begin{figure*}[t]
    \centering

    % First row: (a), (b), (c)
    \begin{subfigure}{0.32\textwidth}
        \centering
        \includegraphics[width=\linewidth]{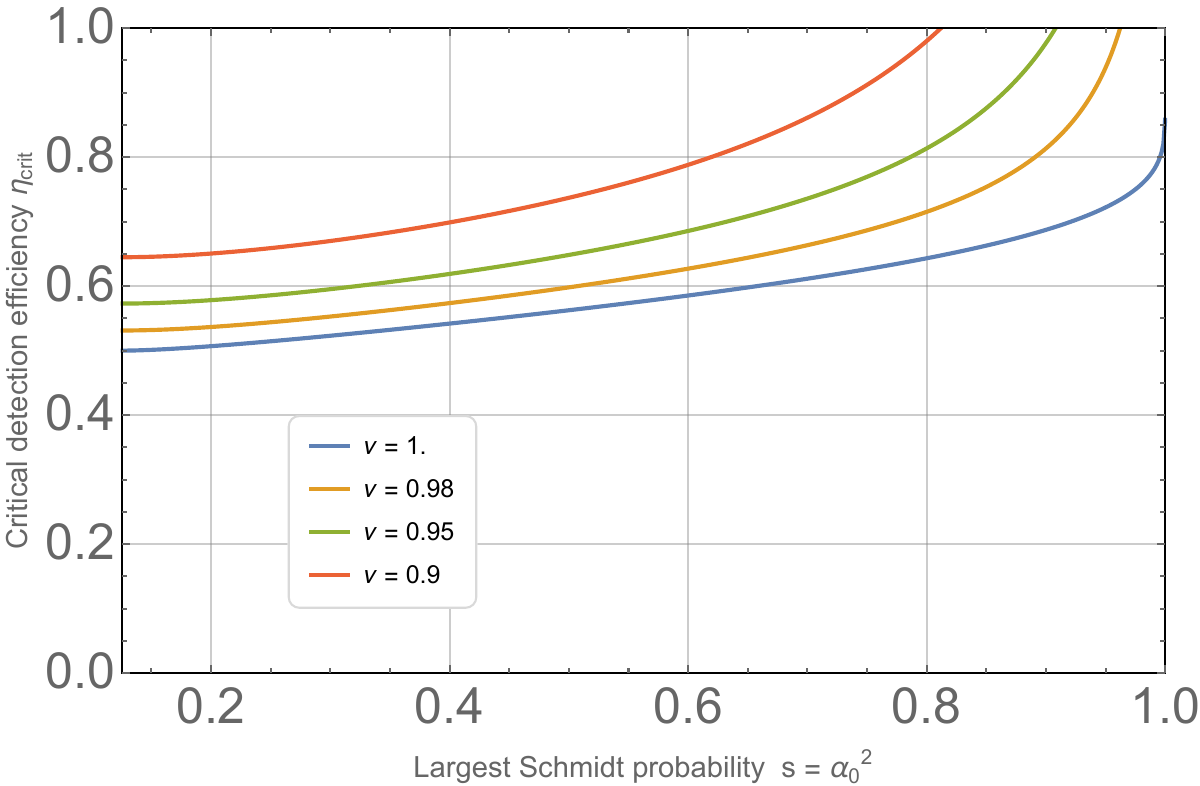}
        \caption{}
        \label{fig:crit_detec_vs_s}
    \end{subfigure}
    \hfill
    \begin{subfigure}{0.32\textwidth}
        \centering
        \includegraphics[width=\linewidth]{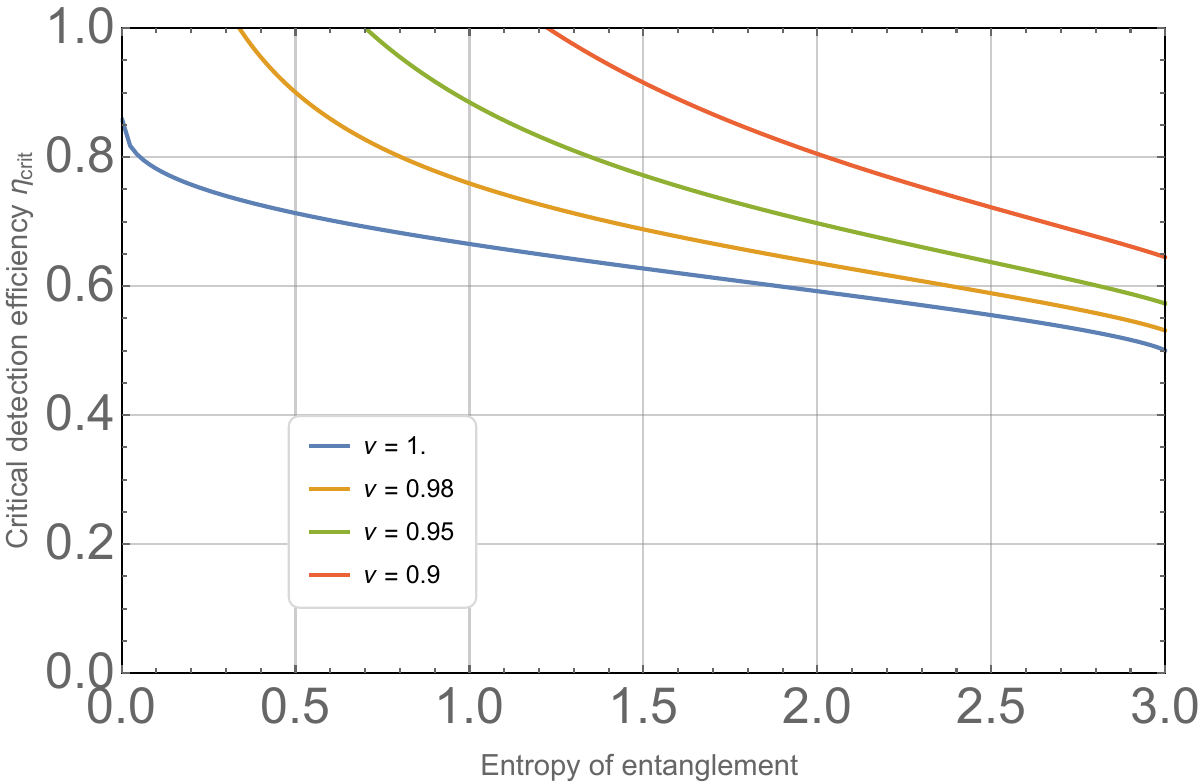}
        \caption{}
        \label{fig:eta_crit_entanglement}
    \end{subfigure}
    \hfill
    \begin{subfigure}{0.32\textwidth}
        \centering
        \includegraphics[width=\linewidth]{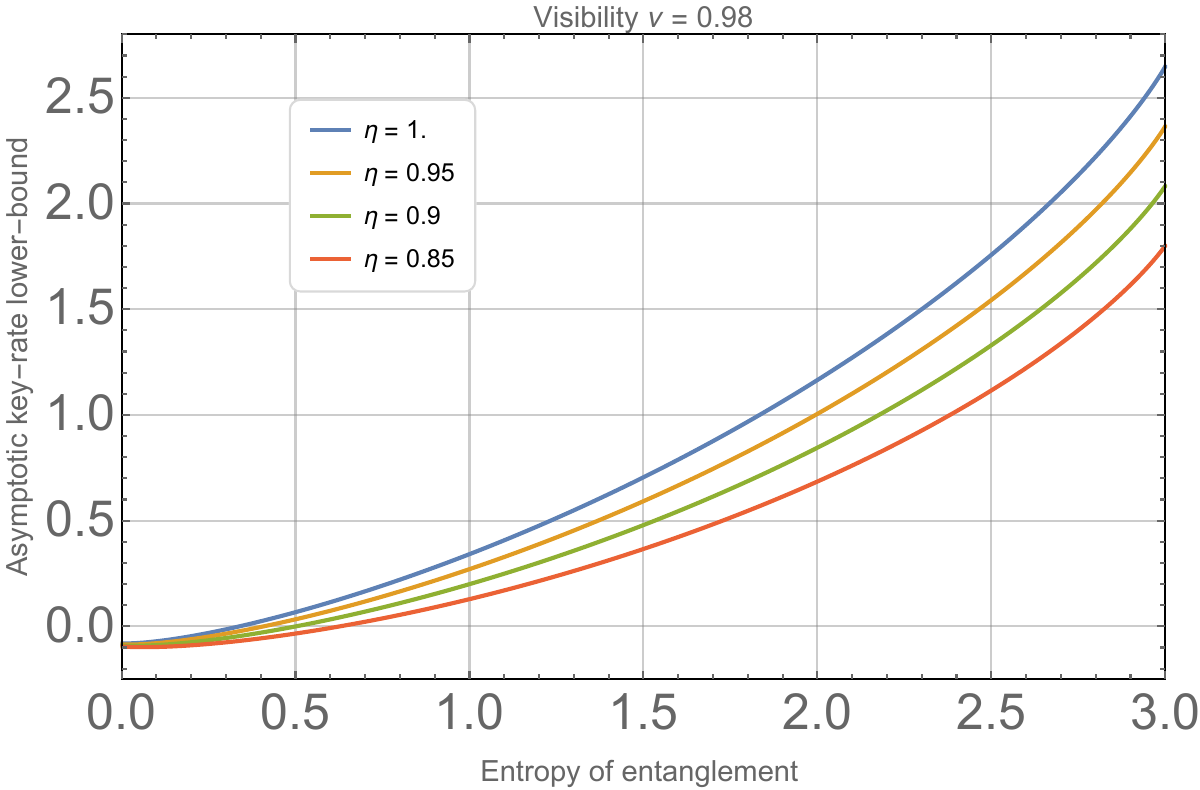}
        \caption{}
        \label{fig:keyrate_entanglement}
    \end{subfigure}

    \vspace{0.3cm}

    % Second row: (d), (e), (f)
    \begin{subfigure}{0.32\textwidth}
        \centering
        \includegraphics[width=\linewidth]{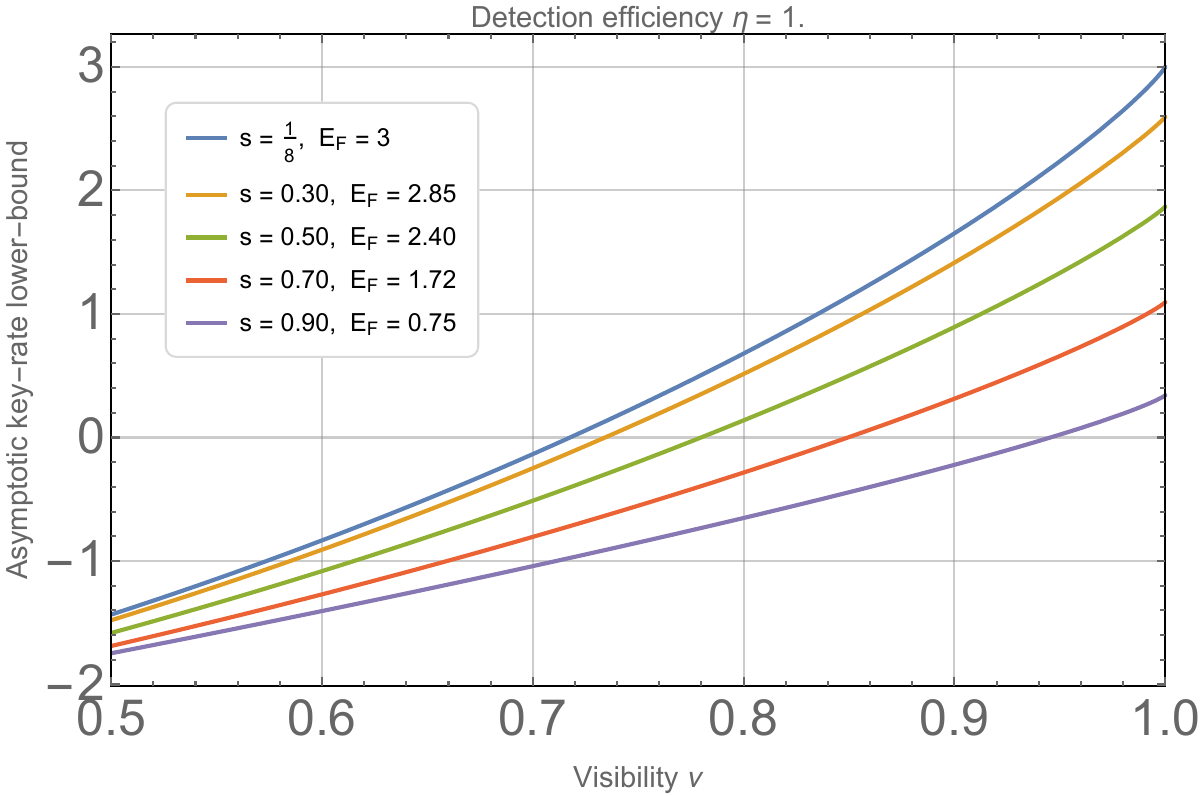}
        \caption{}
        \label{fig:keyrate_visibility_entanglement}
    \end{subfigure}
    \hfill
    \begin{subfigure}{0.32\textwidth}
        \centering
        \includegraphics[width=\linewidth]{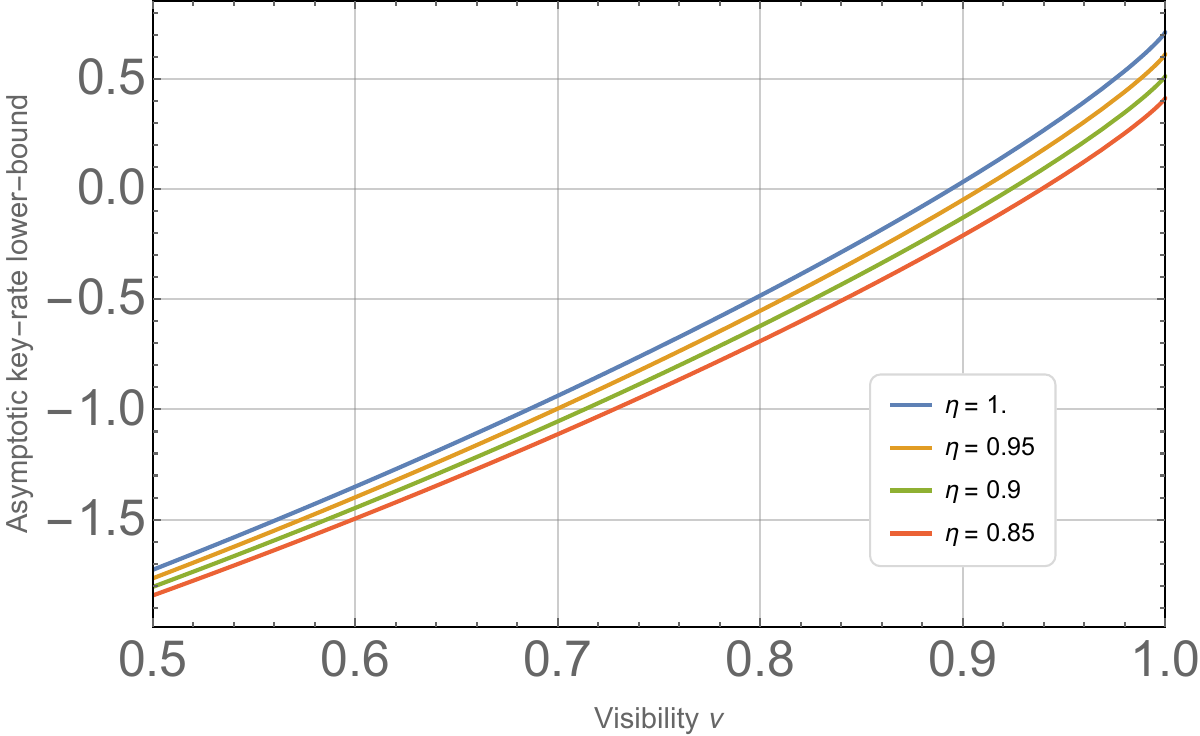}
        \caption{}
        \label{fig:keyrate_vs_visibility}
    \end{subfigure}
    \hfill
    \begin{subfigure}{0.32\textwidth}
        \centering
        \includegraphics[width=\linewidth]{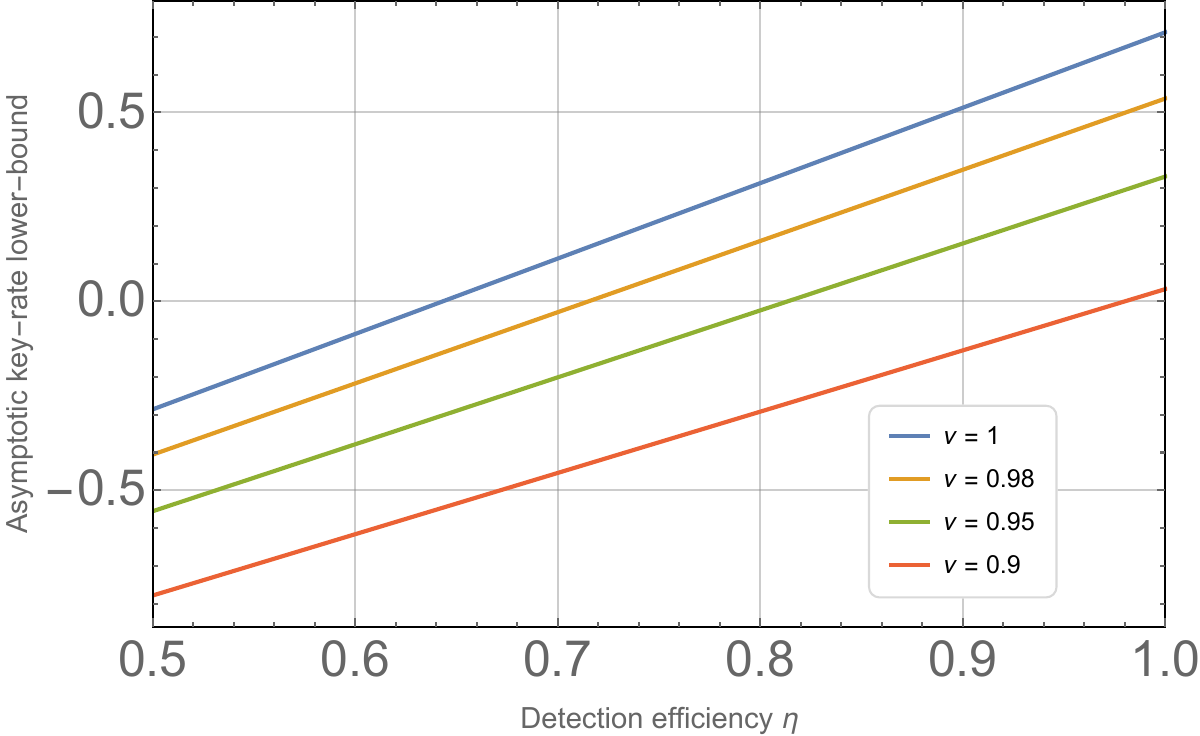}
        \caption{}
        \label{fig:keyrate_detectionefficiency}
    \end{subfigure}

    \caption{\justifying
    Dependence of the EUR-based asymptotic secret-key-rate
    lower bound and critical detection efficiency on the
    state parameters and experimental imperfections.
    (a) Critical detection efficiency versus the largest
    Schmidt probability $s$, for $d=8$.
    (b) Critical detection efficiency versus the entropy
    of entanglement $E_F$, for different visibilities.
    (c) Key-rate lower bound versus $E_F$ for different
    detection efficiencies at $\nu=0.98$.
    (d) Key-rate lower bound versus visibility for states
    with different entanglement entropies at $\eta=1$.
    (e) Key-rate lower bound versus visibility for
    different detection efficiencies.
    (f) Key-rate lower bound versus detection efficiency
    for different visibilities.
    }

    \label{fig:combined_results}
\end{figure*}

\begin{equation}
    p(a,b\mid Z^{(\eta)},Z^*)
    =
    \eta\!\left[\nu\alpha_a^2\delta_{a,b} + \frac{1-\nu}{d^2}\right],
    \label{eq:pabZZ}
\end{equation}
for $a=0,\ldots,d-1$ and any $b$, while for Alice's no-click outcome
\begin{equation}
    p(\emptyset,b\mid Z^{(\eta)},Z^*)
    =
    (1-\eta)\!\left[\nu\lambda_b + \frac{1-\nu}{d}\right].
    \label{eq:pphibZZ}
\end{equation}
On the other hand, when both choose $x=y=1$, for $a=0,1,\ldots,d-1$, and any $b$, we obtain
\begin{equation}
    p(a,b\mid X^{(\eta)},X^*)
    =
    \frac{\eta}{d^2}\!\left[\nu|F_{a-b}|^2 + 1-\nu\right].
    \label{eq:pabXX}
\end{equation}
where $ F_k
    =
    \sum_{i=0}^{d-1}\alpha_i\,\omega^{-ki}$ for any $
    k=0,\ldots,d-1$, while for the no-click outcome,
    \begin{eqnarray}
        p(\emptyset,b\mid X^{(\eta)},X^*)=\frac{(1-\eta)}{d}.
    \end{eqnarray}
    Using these correlations, we can compute the conditional entropies $H(B|A,Z^{(\eta)},Z^*), H(B|A,X^{(\eta)},X^*)$, which allows us to analytically obtain the lower bound to the key-rate for any pure entangled bipartite state.
    \begin{thm}
    \label{thm:noisy}
In the ideal realisation of the 1SDI-QKD protocol described above, let Alice and Bob share the noisy bipartite state \eqref{eq:noisy_state}
such that Alice performs the following measurements $\{Z_d^{(\eta)},X_d^{(\eta)}\}$, while Bob is trusted and measures in $\{Z_d^*,X_d^*\}$. The lower bound to the key-rate for this protocol when the key is extracted from the input $y=0$ is given by
\begin{equation}
    r_{\infty}
    \geq
    \eta\!\left[N_Z + \log_2 d - C_Z - C_X\right]
    -
    N_Z.
    \label{eq:keyRateLinearEta}
\end{equation}
where, 
\begin{eqnarray}
    C_Z
    =
    -\sum_{a=0}^{d-1}\delta_a\left[
        \frac{\nu\alpha_a^2+q}{\delta_a}\log_2\!\left(\frac{\nu\alpha_a^2+q}{\delta_a}\right)\right.
        \nonumber\\ \left.+\,(d-1)\frac{q}{\delta_a}
        \log_2\!\left(\frac{q}{\delta_a}\right)
    \right],
    \label{eq:CZ}
\end{eqnarray}
such that $q = \frac{1-\nu}{d^2}$ and $\delta_a = \nu\alpha_a^2 + \frac{1-\nu}{d}.$ Moreover,
\begin{equation}
    N_Z = -\sum_{a=0}^{d-1}\delta_a\log_2 \delta_a
    \label{eq:NZ},\quad  C_X = -\sum_{k=0}^{d-1}t_k\log_2 t_k
\end{equation}
where $t_k= \frac{\nu|F_k|^2+1-\nu}{d}$.
    \end{thm}
In Appendix A, we show the calculations to obtain the key rate in \eqref{eq:keyRateLinearEta}. Next, we compute the critical detection efficiency required to obtain a secure key between Alice and Bob.

\begin{ucor} The critical efficiency of Alice's detectors to obtain a secure key in the 1SDI-QKD protocol described above for any visibility $\nu$ is given by
\begin{eqnarray}
     \eta_{\mathrm{crit}}^{(\nu)}
    =
    \frac{N_Z}{N_Z + \log_2 d - C_Z - C_X}
\end{eqnarray}
where $C_Z,C_X,N_Z$ are stated in theorem \ref{thm:noisy}. Moreover, the minimum critical detection efficiency is when the source is perfect $(\nu=1)$, and thus
\begin{eqnarray}
     \eta_{\mathrm{crit}}
    =
    \frac{S(\rho_B)}{S(\rho_B) + \log_2 d - \frac{1}{d}\sum_{k=0}^{d-1}|F_k|^2\log_2\frac{|F_k|^2}{d}}
\end{eqnarray}
where $S(\rho_B)$ is the entanglement entropy of $\ket{\psi(\boldsymbol{\alpha})}$ stated in Theorem \ref{thm:pure}.
\end{ucor}
To obtain the above critical detection efficiency, we simply considered that $r_{\infty}=0$. 
While it is straightforward to observe from Theorem \ref{thm:pure} that more entanglement in the reference state implies more key, one can not directly relate the key rate in the realistic case [Theorem \ref{thm:noisy}] to any measure of entanglement. Consequently, to observe the effect of entanglement on key rate, we consider a particular family of bipartite entangled states,
\begin{equation}
    \ket{\psi(s)} = \sqrt{s}\ket{00}+\sqrt{\frac{1-s}{d-1}}\sum_{i=1}^{d-1}\ket{ii}
    \label{eq:specific_class_of_bipartite_ent_states}
\end{equation}
\iffalse
where the Schmidt coefficients are
\begin{equation}
    \lambda_0=s,
    \qquad
    \lambda_i=\frac{1-s}{d-1},
    \quad i=1,\ldots,d-1
\end{equation}
\fi
such that the entropy of entanglement is $E_F(s) =-s\log_2 s -(1-s)\log_2\left(\frac{1-s}{d-1}\right).$ 

In Fig. \ref{fig:combined_results}, we plot the key rate for various parameters of visibility, detection efficiency and entanglement. Figure~\ref{fig:crit_detec_vs_s} shows the critical detection efficiency as a function of the largest Schmidt probability $s$ for different visibilities. For a fixed visibility, the critical detection efficiency increases as $s$ approaches one, corresponding to a reduction in entanglement. Figure~\ref{fig:eta_crit_entanglement} shows the critical detection efficiency as a function of $E_F$ for different values of the visibility. For a fixed visibility, the critical efficiency decreases with increasing entanglement. The effect becomes more prominent as the visibility decreases, showing that reduced entanglement and noise
both increase the detection efficiency required for the present bound to certify a positive key rate. Figure~\ref{fig:keyrate_entanglement} shows the asymptotic key-rate bound as a function of $E_F$ at fixed visibility $\nu=0.98$ and for different detection efficiencies. The bound
increases with the amount of entanglement for all values of $\eta$ considered. Detection loss lowers the achievable bound, while the largest values are obtained close to the maximally entangled state. Figure~\ref{fig:keyrate_visibility_entanglement} shows the asymptotic key-rate bound as a function of the visibility $\nu$ for states with different amounts of entanglement, taking
$\eta=1$. We consider $s=0.125,0.30,0.50,0.70$, and $0.90$, which correspond to $E_F=3,2.85,2.40,1.72$, and $0.75$, respectively. For all the states considered, the bound increases with visibility, with the maximally entangled state giving the largest key-rate bound throughout the range shown.

The effect of decreasing visibility also depends on the amount of entanglement. The more entangled states show a larger absolute change in the key-rate bound as $\nu$ is varied, whereas the curves become flatter as the entanglement is reduced. In particular, the maximally entangled state gives the largest bound over the full range of visibilities considered. Figure~\ref{fig:combined_results}\subref{fig:keyrate_vs_visibility} shows the dependence of the key-rate lower bound on visibility for different detection efficiencies. For a fixed visibility, reducing the detection efficiency lowers the bound, while increasing visibility improves it for all the detection efficiencies considered. Figure~\ref{fig:keyrate_detectionefficiency} shows the key-rate lower bound as a function of detection efficiency for different visibilities. For a fixed visibility, the bound increases linearly with $\eta$, while decreasing $\nu$ shifts the curves downwards. Consequently, a higher detection efficiency is required to certify a positive key rate at lower visibilities.

% -------------------------------------------------------------------

{\it{Discussions---}} While it is known that any pure entangled bipartite state is steerable (or nonlocal) \cite{wiseman2006steering, Gisin1991}, it was not yet known whether such states are useful for nonlocality-based QKD protocols. As a matter of fact, no such protocol has been constructed beyond qubits, with the exception of \cite{riveradean2025device}, which demonstrates DIQKD for a qutrit-qutrit entangled state and \cite{1SDIQKD5} that demonstrates 1SDI-QKD using a maximally entangled state of arbitrary dimension. Here, we demonstrated that any pure bipartite entangled state is useful for generating a secure key in nonlocality-based QKD protocols. We then considered realistic sources and detectors and obtained the key-rate for any such state. We observed that the key-rate is robust against entanglement (when it is not low enough); thus, a maximally entangled state, which is a costly resource, is not required to obtain a secure key among distant parties. Moreover, an important observation we make here is that, unlike \cite{Eberhard1993, Vallone_2013}, the critical detection efficiency does not improve when using non-maximally entangled states over maximally entangled states for the standard nonlocality-based key distribution protocols.

Several follow-up problems arise from our work. Finding tighter lower bounds for arbitrary pure entangled states would certainly be interesting. Extending the protocol where Bob is not fully trusted, similar to distrustful quantum steering \cite{disquantsteer1,disquantsteer2}, will elevate 1SDI-QKD protocols towards DI-QKD. As of now, our approach in the realistic case does not apply to DI-QKD protocols, as the relation \cite{Berta2010} is only valid when the measurements on one side are known to be rank-one projective. Consequently, extending this to the DI regime would resolve some major long-standing problems in DI-QKD.

{\it{Acknowledgements---}}The authors thank Will McCutcheon, Suraj Goel and Mehul Malik for insightful discussions. 
B.G. acknowledges financial support from the UK Engineering and Physical Sciences Research Council (EPSRC) (Grant Nos. EP/Z533208/1 and EP/W003252/1).
C.R. acknowledges support by the IRA Programme, project
no.~FENG.02.01-IP.05-0006/23, financed by the FENG program 2021--2027, Priority FENG.02, Measure FENG.02.01, with the support of the FNP.

%\bibliographystyle{apsrev4-2}
%\bibliography{references}

\onecolumngrid
\appendix
\label{appendix}
\setcounter{thm}{1}

\section{Key-rate in the realistic 1SDI-QKD protocol}
\begin{thm}
        \label{thm:noisyapp}
In the ideal realisation of the 1SDI-QKD protocol described above, let Alice and Bob share the noisy bipartite state \eqref{eq:noisy_state}
such that Alice performs the following measurements $\{Z_d^{(\eta)},X_d^{(\eta)}\}$, while Bob is trusted and measures in $\{Z_d^*,X_d^*\}$. The lower bound to the key-rate for this protocol when the key is extracted from the input $y=0$ is given by
\begin{equation}
    r_{\infty}
    \geq
    \eta\!\left[N_Z + \log_2 d - C_Z - C_X\right]
    -
    N_Z.
    \label{eq:keyRateLinearEtaapp}
\end{equation}
where, 
\begin{eqnarray}
    C_Z
    =
    -\sum_{a=0}^{d-1}\delta_a\left[
        \frac{\nu\alpha_a^2+q}{\delta_a}\log_2\!\left(\frac{\nu\alpha_a^2+q}{\delta_a}\right)+\,(d-1)\frac{q}{\delta_a}
        \log_2\!\left(\frac{q}{\delta_a}\right)
    \right],
    \label{eq:CZ}
\end{eqnarray}
such that $q = \frac{1-\nu}{d^2}$ and $\delta_a = \nu\alpha_a^2 + \frac{1-\nu}{d}.$ Moreover,
\begin{equation}
    N_Z = -\sum_{a=0}^{d-1}\delta_a\log_2 \delta_a
    \label{eq:NZ},\quad  C_X = -\sum_{k=0}^{d-1}t_k\log_2 t_k
\end{equation}
where $t_k= \frac{\nu|F_k|^2+1-\nu}{d}$.
    \end{thm}
For simplicity, the calculations are divided into subsections. 
\subsection{Probability Distributions for the \texorpdfstring{$Z,Z^*$}{Z,Z*} Measurements}
 
In the $Z$ basis, Alice measures in the computational basis $\{\ket{a}\!\bra{a}\}$ and
Bob measures in the complex-conjugate computational basis $\{\ket{a^*}\!\bra{a^*}\}$.
For rounds in which Alice registers a detection, the joint probability distribution is
\begin{equation}
    p(a,b\mid Z,Z^*)
    =
    \eta\!\left[\nu\lambda_a\delta_{a,b} + \frac{1-\nu}{d^2}\right],
    \qquad a=0,\ldots,d-1.
    \label{eq:pabZZ}
\end{equation}
The Kronecker delta $\delta_{a,b}$ encodes the perfect correlations of the noiseless
state, while the term $(1-\nu)/d^2$ accounts for depolarizing noise.
For Alice's no-click outcome,
\begin{equation}
    p(\emptyset,b\mid Z,Z^*)
    =
    (1-\eta)\!\left[\nu\lambda_b + \frac{1-\nu}{d}\right].
    \label{eq:pphibZZ}
\end{equation}
These two contributions are correctly normalised: summing $p(a,b\mid Z,Z^*)$ over all
$a,b$ yields $\eta$, while summing $p(\emptyset,b\mid Z,Z^*)$ over $b$ yields $1-\eta$.
The marginal probability for a detected Alice outcome is
\begin{equation}
    p(a\mid Z,Z^*) = \eta\,r_a,
\end{equation}
and the conditional distributions are
\begin{equation}
    p(b\mid a,Z,Z^*)
    =
    \frac{\nu\lambda_a\delta_{a,b}+q}{r_a},
    \qquad
    p(b\mid\emptyset,Z,Z^*)
    =
    r_b.
    \label{eq:pbGivenaZZ}
\end{equation}
In particular, $p(b=a\mid a,Z,Z^*)=(\nu\lambda_a+q)/r_a$, while for each $b\neq a$ one
has $p(b\mid a,Z,Z^*)=q/r_a$.
 
% -------------------------------------------------------------------
\subsection{Conditional Entropy for the \texorpdfstring{$Z,Z^*$}{Z,Z*} Measurements}
 
Separating the contributions from detected and no-click events, the conditional entropy is
\begin{equation}
    H(B\mid A,Z,Z^*) = \eta\,C_Z + (1-\eta)\,N_Z,
    \label{eq:HZZCompact}
\end{equation}
where the detected-event contribution is
\begin{align}
    C_Z
    =
    -\sum_{a=0}^{d-1}r_a\Bigg[
        \frac{\nu\lambda_a+q}{r_a}
        \log_2\!\left(\frac{\nu\lambda_a+q}{r_a}\right)
        +\,(d-1)\frac{q}{r_a}
        \log_2\!\left(\frac{q}{r_a}\right)
    \Bigg],
    \label{eq:CZ}
\end{align}
and the no-click contribution is
\begin{equation}
    N_Z = -\sum_{b=0}^{d-1}r_b\log_2 r_b.
    \label{eq:NZ}
\end{equation}
Expanding in full,
\begin{align}
    H(B\mid A,Z,Z^*)
    =
    -\eta\sum_{a=0}^{d-1}
    \left(\nu\lambda_a+q\right)
    \log_2\!\frac{\nu\lambda_a+q}{r_a}
   -\eta(d-1)\sum_{a=0}^{d-1}
    q\log_2\!\frac{q}{r_a}
   -(1-\eta)\sum_{b=0}^{d-1}
    r_b\log_2 r_b.
    \label{eq:HZZExplicit}
\end{align}
We adopt the convention $0\log_2 0 = 0$ throughout.
 
% -------------------------------------------------------------------
\subsection{Probability Distributions for the \texorpdfstring{$X,X^*$}{X,X*} Measurements}
 
In the $X$ basis, Alice measures in the Fourier basis
\begin{equation}
    \ket{x_a}
    =
    \frac{1}{\sqrt{d}}\sum_{j=0}^{d-1}\omega^{aj}\ket{j},
\end{equation}
while Bob measures in the complex-conjugate Fourier basis
\begin{equation}
    \ket{x_b^*}
    =
    \frac{1}{\sqrt{d}}\sum_{j=0}^{d-1}\omega^{-bj}\ket{j}.
\end{equation}
We define the Fourier amplitudes
\begin{equation}
    F_k
    =
    \sum_{i=0}^{d-1}\alpha_i\,\omega^{-ki},
    \qquad
    k=0,\ldots,d-1,
    \label{eq:Fk}
\end{equation}
which satisfies $\sum_{k=0}^{d-1}|F_k|^2 = d$.
For detected events at Alice, the joint probability distribution is
\begin{equation}
    p(a,b\mid X,X^*)
    =
    \frac{\eta}{d^2}\!\left[\nu|F_{a-b}|^2 + 1-\nu\right].
    \label{eq:pabXX}
\end{equation}
For Alice's no-click outcome, $p(\emptyset,b\mid X,X^*)=(1-\eta)/d$.
The marginal over $b$ is uniform, $p(a\mid X,X^*)=\eta/d$, for every detected outcome $a$.
Defining
\begin{equation}
    t_k
    =
    \frac{\nu|F_k|^2+1-\nu}{d},
    \label{eq:tk}
\end{equation}
which satisfies $\sum_{k=0}^{d-1}t_k=1$, the conditional distribution is
$p(b\mid a,X,X^*)=t_{a-b}$.
Because changing $a$ merely cyclically permutes the values $\{t_{a-b}\}$, the conditional
entropy for a detected event is independent of $a$.
For the no-click outcome, $p(b\mid\emptyset,X,X^*)=1/d$.
 
% -------------------------------------------------------------------
\subsection{Conditional Entropy for the \texorpdfstring{$X,X^*$}{X,X*} Measurements}
 
The conditional entropy for the $X$ basis takes the compact form
\begin{equation}
    H(B\mid A,X,X^*)
    =
    \eta\,C_X + (1-\eta)\log_2 d,
    \label{eq:HXXCompact}
\end{equation}
where
\begin{equation}
    C_X = -\sum_{k=0}^{d-1}t_k\log_2 t_k
    \label{eq:CX}
\end{equation}
is the Shannon entropy of the distribution $\{t_k\}$.
Substituting the definition of $t_k$ explicitly, we obtain
\begin{align}
    H(B\mid A,X,X^*)
    ={}
    -\eta\sum_{k=0}^{d-1}
    \frac{\nu|F_k|^2+1-\nu}{d}
    \log_2\!\left[\frac{\nu|F_k|^2+1-\nu}{d}\right]
   +\,(1-\eta)\log_2 d.
    \label{eq:HXXExplicit}
\end{align}
The no-click term $(1-\eta)\log_2 d$ reflects the fact that, when Alice fails to detect,
Bob's outcome is uniformly distributed over $d$ values.
 
% -------------------------------------------------------------------
\subsection{Key Rate via Entropic Uncertainty Relations}
 
The asymptotic secret key rate for the two-basis protocol is bounded below by the relation,
\begin{equation}
    r_{\infty}
    \geq
    \log_2 d
    -
    H(B\mid A,Z,Z^*)
    -
    H(B\mid A,X,X^*).
    \label{eq:keyRate}
\end{equation}
Substituting the expressions derived above and collecting terms proportional to $\eta$,
\begin{equation}
    r_{\infty}
    \geq
    \eta\!\left[N_Z + \log_2 d - C_Z - C_X\right]
    -
    N_Z.
    \label{eq:keyRateLinearEta}
\end{equation}
This expression is affine in the detection efficiency $\eta$, with slope
$N_Z + \log_2 d - C_Z - C_X$ and intercept $-N_Z$, making the key rate's behaviour
under loss transparent: it decreases linearly as $\eta$ is reduced, eventually vanishing
at the critical efficiency derived below.
 
% -------------------------------------------------------------------
\subsection{Critical Detection Efficiency}
 
The critical detection efficiency $\eta_{\mathrm{crit}}$ is the minimum efficiency at
which a positive key rate is achievable, defined by $r_{\infty}\left(\eta_{\mathrm{crit}}^{(\nu)}\right)=0$.
Solving~\eqref{eq:keyRateLinearEta},
\begin{equation}
    \eta_{\mathrm{crit}}^{(\nu)}
    =
    \frac{N_Z}{N_Z + \log_2 d - C_Z - C_X}.
    \label{eq:etaCritical}
\end{equation}
\iffalse
A strictly positive key can be extracted at some physical efficiency $\eta\leq 1$ only
when
\begin{equation}
    C_Z + C_X < \log_2 d,
\end{equation}
i.e.\ the total conditional entropy across the two measurement bases must fall below
the maximum entropy $\log_2 d$.
When this condition holds, the secure operating region is $\eta > \eta_{\mathrm{crit}}$.
\fi
The critical efficiency depends on the state parameters $\{\lambda_i\}$ and the
visibility $\nu$ through $C_Z$, $C_X$, and $N_Z$.

\end{document}